\documentclass{sig-alternate}

\newtheorem{theorem}{Theorem}

\newtheorem{lemma}[theorem]{Lemma}

\def \CC   {{\cal C}}

\def \SS   {{\cal S}}

\begin{document}
%

\conferenceinfo{SPAA'05,} {July 18--20, 2005, Las Vegas, Nevada,
USA.}

\CopyrightYear{2005}

\crdata{1-58113-986-1/05/0007}


\title{Admission Control to Minimize Rejections and
Online Set Cover with Repetitions}
%
%

\numberofauthors{3}
%

\author{
%
\alignauthor Noga Alon \titlenote{Research supported in part by a
grant from the Israel Science Foundation, and by the Hermann
Minkowski Minerva Center for Geometry at Tel Aviv University.}\\
       \affaddr{Schools of Mathematics and Computer Science}\\
       \affaddr{Tel-Aviv University}\\
       \affaddr{Tel-Aviv, 69978, Israel}\\
       \email{noga@math.tau.ac.il}
\alignauthor Yossi Azar \titlenote{Research supported in part by
the Israel Science Foundation and by the German-Israeli Foundation.}\\
       \affaddr{School of Computer Science}\\
       \affaddr{Tel-Aviv University}\\
       \affaddr{Tel-Aviv, 69978, Israel}\\
       \email{azar@tau.ac.il}
\alignauthor Shai Gutner \titlenote{This paper forms part of a
Ph.D. thesis written by the author under the supervision of Prof.
N. Alon and Prof. Y. Azar in Tel Aviv University.}\\
       \affaddr{School of Computer Science}\\
       \affaddr{Tel-Aviv University}\\
       \affaddr{Tel-Aviv, 69978, Israel}\\
       \email{gutner@tau.ac.il}
}
\maketitle

\begin{abstract}

We study the admission control problem in general networks.
Communication requests arrive over time, and the online algorithm
accepts or rejects each request while maintaining the capacity
limitations of the network. The admission control problem has been
usually analyzed as a benefit problem, where the goal is to devise
an online algorithm that accepts the maximum number of requests
possible. The problem with this objective function is that even
algorithms with optimal competitive ratios may reject almost all
of the requests, when it would have been possible to reject only a
few. This could be inappropriate for settings in which rejections
are intended to be rare events.

In this paper, we consider preemptive online algorithms whose goal
is to minimize the number of rejected requests. Each request
arrives together with the path it should be routed on. We show an
$O(\log^2 (mc))$-competitive randomized algorithm for the weighted
case, where $m$ is the number of edges in the graph and $c$ is the
maximum edge capacity. For the unweighted case, we give an $O(\log
m \log c)$-competitive randomized algorithm. This settles an open
question of Blum, Kalai and Kleinberg raised in \cite{BlKaKl01}.
We note that allowing preemption and handling requests with given
paths are essential for avoiding trivial lower bounds.

The admission control problem is a generalization of the online
set cover with repetitions problem, whose input is a family of $m$
subsets of a ground set of $n$ elements. Elements of the ground
set are given to the online algorithm one by one, possibly
requesting each element a multiple number of times. (If each
element arrives at most once, this corresponds to the online set
cover problem.) The algorithm must cover each element by different
subsets, according to the number of times it has been requested.

We give an $O(\log m \log n)$-competitive randomized algorithm for
the online set cover with repetitions problem. This matches a
recent lower bound of $\Omega(\log m \log n)$ given by Feige and
Korman for the competitive ratio of any randomized {\em
polynomial} time algorithm, under the $BPP \neq NP$ assumption.
Given any constant $\epsilon > 0$, an $O(\log m \log
n)$-competitive deterministic bicriteria algorithm is shown that
covers each element by at least $(1-\epsilon)k$ sets, where $k$ is
the number of times the element is covered by the optimal
solution.


\end{abstract}





\section{Introduction}\label{sec:intro}

We study the admission control problem in general graphs with edge
capacities. An online algorithm can receive a sequence of
communications requests on a virtual path, that may be accepted or
rejected, while staying within the capacity limitations.

This problem has typically been studied as a benefit problem. This
means that the online algorithm has to be competitive with respect
to the number of accepted requests. A problem with this objective
function is that in some cases an online algorithm with a good
competitive ratio may reject the vast majority of the requests,
whereas the optimal solution rejects only a small fraction of
them.

In this paper we consider the goal of minimizing the number of
rejected requests, which was first studied in \cite{BlKaKl01}.
This approach is suitable for applications in which rejections are
intended to be rare events. A situation in which a significant
fraction of the requests is rejected even by the optimal solution
means that the network needs to be upgraded.

We consider preemptive online algorithms for the admission control
problem. Allowing preemption is necessary for achieving reasonable
bounds for the competitive ratio. Each request arrives together
with the path it should be routed on. The admission control
algorithm decides whether to accept or reject it. An online
algorithm for both admission control and routing easily admits a
trivial lower bound \cite{BlKaKl01}.

\textbf{The admission control to minimize rejections problem.} We
now formally define the admission control problem. The input
consist of the following:
\begin{itemize}
    \item A directed graph $G=(V,E)$, where $|E|=m$.
    Each edge $e$ has an integer capacity $c_e>0$.
    We denote $c=max_{e \in E} c_e$.
    \item A sequence of requests $r_1,r_2,\ldots,$ each of which is a simple path in the graph.
    Every request $r_i$ has a cost $p_i>0$ associated with it.
\end{itemize}

A feasible solution for the problem must assure that for every
edge $e$, the number of accepted requests whose paths contain $e$
is at most its capacity $c_e$. The goal is to find a feasible
solution of minimum cost of the rejected requests. The online
algorithm is given requests one at a time, and must decide whether
to accept or reject each request. It is also allowed to preempt a
request, i.e. to reject it after already accepting it, but it
cannot accept a request after rejecting it.

Let $OPT$ be a feasible solution having minimum cost $C_{OPT}$. An
algorithm is $\beta$-competitive if the total cost of the requests
rejected by this algorithm is at most $\beta C_{OPT}$.

\textbf{Previous results for admission control.} Tight bounds were
achieved for the admission control problem, where the goal is to
maximize the number of accepted requests. Awerbuch, Azar and
Plotkin \cite{AAP93} provide an $O(\log n)$-competitive algorithm
for general graphs. For the admission control problem on a tree,
$O(\log d)$-competitive randomized algorithms appear in
\cite{ABFR94,AGLR94}, where $d$ is the diameter of the tree. Adler
and Azar presented a constant-competitive preemptive algorithm for
admission control on the line \cite{AdAz03}.

The admission control to minimize rejections problem was studied
by Blum, Kalai and Kleinberg in \cite{BlKaKl01}, where two
deterministic algorithms with competitive ratios of $O(\sqrt{m})$
and $c+1$ are given ($m$ is the number of edges in the graph and
$c$ is the maximum capacity). They raised the question of whether
an online algorithm with polylogarithmic competitive ratio can be
obtained.

We note that one can combine an algorithm for maximizing
throughput of accepted requests and an algorithm for minimizing
rejections and get one algorithm which achieves both
simultaneously with slightly degrading the competitive ratio
\cite{ABM03,BM04}.

In this paper we show that the admission control to minimize
rejections problem is a generalization of the online set cover
with repetitions problem described below:

\textbf{The online set cover with repetitions problem.} The online
set cover problem is defined as follows: Let $X$ be a ground set
of $n$ elements, and let $\SS$ be a family of subsets of $X$,
$|\SS|=m$. Each $S \in \SS$ has a non-negative cost associated
with it. An adversary gives elements to the algorithm from $X$ one
by one. Each element of $X$ can be given an arbitrary number of
times, not necessarily consecutively. An element should be covered
by a number of sets which is equal to the number of times it
arrived. We assume that the elements of $X$ and the members of
$\SS$ are known in advance to the algorithm, however, the elements
given by the adversary are not known in advance. The objective is
to minimize the cost of the sets chosen by the algorithm.

\textbf{Previous results for online set cover.} The offline
version of the set cover problem is a classic NP-hard problem that
was studied extensively, and the best approximation factor
achievable for it in polynomial time (assuming $P \neq NP$) is
$\Theta(\log n)$ \cite{C79,F98}. The basic online set cover
problem, where repetitions are not allowed, was studied in
\cite{AAABN03,FK05}. A different variant of the problem, dealing
with maximum benefit, is presented in \cite{AAFL96}. An $O(\log m
\log n)$-competitive deterministic algorithm for the online set
cover problem was given by \cite{AAABN03} where $n$ is the number
of elements and $m$ is the number of sets. A lower bound of
$\Omega(\frac{\log m \log n}{\log \log m + \log \log n})$ was also
shown for any deterministic online algorithm. A recent result of
Feige and Korman \cite{FK05} establishes a lower bound of
$\Omega(\log m \log n)$ for the competitive ratio of any
randomized {\em polynomial} time algorithm for the online set
cover problem, under the $BPP \neq NP$ assumption. They also prove
the same lower bound for any deterministic {\em polynomial} time
algorithm, under the $P \neq NP$ assumption.

\textbf{Our results.} The main result we give in this paper is an
$O(\log^2 (mc))$-competitive randomized algorithm for the
admission control to minimize rejections problem. This settles the
open question raised by Blum et al. \cite{BlKaKl01}. For the
unweighted case, where all costs are equal to 1, we slightly
improve this bound and give an $O(\log m \log c)$-competitive
randomized algorithm,

We present a simple reduction between online set cover with
repetitions and the admission control to minimize rejections
problem. This implies an $O(\log^2 (mn))$-competitive randomized
algorithm for the online set cover with repetitions problem. For
the unweighted case (all costs are equal to $1$), we get an
$O(\log m \log n)$-competitive randomized algorithm. This matches
the lower bound of $\Omega(\log m \log n)$ given by Feige and
Korman. Their results also imply a lower bound of $\Omega(\log m
\log c)$ for the competitive ratio of any randomized {\em
polynomial} time algorithm for the admission control to minimize
rejections problem (assuming $BPP \neq NP$).

The derandomization techniques used in \cite{AAABN03} for the
online set cover problem do not seem to apply here. This is why we
also consider the bicriteria version of the online set cover with
repetition problem. For a given constant $\epsilon > 0$, the
online algorithm is required to cover each element by a fraction
of $1-\epsilon$ times the number of its appearances. Specifically,
at any point of time, if an element has been requested $k$ times
so far, then the optimal solution covers it by $k$ different sets,
whereas the online algorithm covers it by $(1-\epsilon)k$
different sets. We give an $O(\log m \log n)$-competitive
deterministic bicriteria algorithm for this problem.

\textbf{Techniques.} The techniques we use follow those of
\cite{AAABN03,AAABN04} together with some new ideas. We start with
an online fractional solution which is monotone increasing during
the algorithm. Then, the fractional solution is converted into a
randomized algorithm. Interestingly, to get a deterministic
bicriteria algorithm we use a different fractional algorithm than
the one used for the randomized algorithm.


\section{Fractional algorithm for \\ admission control}\label{sec:frac}
In this section we describe a fractional algorithm for the
problem. A fractional algorithm is allowed to reject a fraction of
a request $r_i$. We use a weight $f_i$ for this fraction.
Specifically, if $0 \leq f_i < 1$, we reject with percentage of
precisely $f_i$. If $f_i \geq 1$, then the request is completely
rejected. At any stage of the fractional algorithm we will use the
following notation:
\begin{itemize}
    \item $REQ_e$ will denote the set of requests that arrived
    so far whose paths contain the edge $e$.
    \item $REQ$ will denote $\bigcup_{e \in E} REQ_e$.
    \item $ALIVE_e$ will denote the requests from $REQ_e$
    that have not been fully rejected (requests $r_i$ for which $f_i < 1$).
    \item $n_e$ will denote the excess of edge $e$ caused by the
    requests in $ALIVE_e$.
    $$n_e = |ALIVE_e| - c_e$$
\end{itemize}
The requirement from a fractional algorithm is that for every edge
$e$,
$$ \sum_{i \in ALIVE_e} f_i \geq n_e $$
The cost associated with a fractional algorithm is defined to be
$\sum_{i \in REQ} \min \{f_i,1\} p_i$.

We will now describe an $O(\log (mc))$-competitive algorithm for
the problem, even versus a fractional optimum. The cost of the
optimal fractional solution, $C_{OPT}$ is denoted by $\alpha$.

We may assume, by doubling, that the value of $\alpha$ is known up
to a factor of $2$. To determine the initial value of $\alpha$ we
look for the first time in which we must reject a request from an
edge $e$. We can start guessing $\alpha = min_{i \in REQ_e} p_i$,
and then run the algorithm with this bound on the optimal
solution. If it turns out that the value of the optimal solution
is larger than our current guess for it, (that is, the cost
exceeds $\Theta(\alpha \log (mc))$), then we "forget" about all
the request fractions rejected so far, update the value of
$\alpha$ by doubling it, and continue. We note that the cost of
the request fractions that we have "forgotten" about can increase
the cost of our solution by at most a factor of $2$, since the
value of $\alpha$ was doubled in each step.

We thus assume that $\alpha$ is known. Denote by $R_{big}$ the
requests with cost exceeding $2 \alpha$. The optimal fractional
solution can reject a total fraction of at most $1/2$ out of the
requests of $R_{big}$. Hence, when an edge is requested more than
its capacity, the fractional optimum must reject a total fraction
of at least $1/2$ out of the requests not in $R_{big}$ whose paths
contain the edge. By doubling the fraction of rejection for all
the requests not in $R_{big}$ (keeping fractions to be at most
$1$) and completely accepting all the requests in $R_{big}$, we
get a feasible fractional solution whose cost is at most twice the
optimum. Hence, the online algorithm can always completely accept
requests of cost exceeding $2 \alpha$ (and adjust the edge
capacities $c_e$ accordingly).

Denote by $R_{small}$ the requests with cost at most
$\alpha/(mc)$. We claim that we can completely reject all the
requests from $R_{small}$. For each edge $e$, the optimal solution
can accept a total fraction of at most $c$ out of the requests
whose paths contain the edge $e$, and therefore it can accept a
total fraction of at most $mc$ requests. The fractions of requests
accepted out of $R_{small}$ have total cost at most $mc \cdot
\alpha/(mc) = \alpha$. It follows that the optimal solution pays
at least $cost(R_{small})-\alpha$ for the fractions of requests
out of $R_{small}$ that it rejected. Therefore, the online
algorithm can reject all the requests in $R_{small}$ and pay
$cost(R_{small})$. If $cost(R_{small}) < 2\alpha$, then this adds
only $O(\alpha)$ to the cost of the online algorithm. If
$cost(R_{small}) \geq 2\alpha$, then $cost(R_{small}) \leq
2(cost(R_{small})-\alpha)$, so the online algorithm is
$2$-competitive with respect to the requests in $R_{small}$.

By the above arguments, all the requests of cost smaller than
$\alpha/(mc)$ or greater than $2 \alpha$ are rejected immediately
or accepted permanently (edge capacities are decreased in this
case), respectively. An algorithm needs to handle only requests
of cost between $\alpha/(mc)$ and $2 \alpha$. We normalize the
costs so that the minimum cost is $1$ and the maximum cost is $g
\leq 2mc$, and fix $\alpha$ appropriately.

The algorithm maintains a weight $f_i$ for each request $r_i$. The
weights can only increase during the run of the algorithm.
Initially $f_i=0$ for all the requests. Assume now that the
algorithm receives a request $r_i$ for a path of cost $p_i$. For
each edge $e$, we update $REQ_e$, $ALIVE_e$ and $n_e$ according to
the definitions given above. The following is performed for all
the edges $e$ of the path of $r_i$, in an arbitrary order.

\begin{enumerate}
    \item If $\sum_{i \in ALIVE_e} f_i \geq n_e$, then do
    nothing.
    \item Else, while $\sum_{i \in ALIVE_e} f_i < n_e$, perform a {\em weight augmentation}:
    \begin{enumerate}
        \item For each $i \in ALIVE_e$, if $f_i=0$, then set
        $f_i=1/(gc)$.\label{alg1}
        \item For each $i \in ALIVE_e$, $f_i \gets
        f_i(1+\frac{1}{n_e p_i})$.\label{alg2}
        \item Update $ALIVE_e$ and $n_e$.
    \end{enumerate}
\end{enumerate}

Note that the fractional algorithm starts with all weights equal
to zero. This is necessary, since the online algorithm must reject
$0$ requests in case the optimal solution rejects $0$ requests.
Hence, the algorithm is competitive for $\alpha = 0$, and from now
on we assume without loss of generality that $\alpha > 0$. In the
following we analyze the performance of the algorithm.

\begin{lemma}\label{frac_steps}
The total number of weight augmentations performed during the
algorithm is at most $O(\alpha \log (gc))$.
\end{lemma}

\begin{proof}
Consider the following potential function:
$$ \Phi = \prod_{i \in REQ} {max\{f_i,1/(gc)\}}^{f^*_i p_i}$$
where $f^*_i$ is the weight of the request $r_i$ in the optimal
fractional solution. We now show three properties of $\Phi$:
\begin{itemize}
    \item The initial value of the potential function is:
    $(gc)^{-\alpha}$.
    \item The potential function never exceeds $2^\alpha$.
    \item In each weight augmentation step, the potential function
    is multiplied by at least $2$.
\end{itemize}
The first two properties follow directly from the initial value
and from the fact that no request gets a weight of more than
$1+1/p_i \leq 2$. Consider an iteration in which the adversary
gives a request $r_i$ with cost $p_i$. Now suppose that a weight
augmentation is performed for an edge $e$. We must have $\sum_{i
\in ALIVE_e} f^*_i \geq n_e$ since the optimal solution must
satisfy the capacity constraint. Thus, the potential function is
multiplied by at least:
$$
\prod_{i \in ALIVE_e} \left(1+\frac{1}{n_e p_i}\right)^{f^*_i p_i}
\geq \prod_{i \in ALIVE_e} \left(1 + \frac{1}{n_e}\right)^{f^*_i}
\geq 2
$$
The first inequality follows since for all $x \geq 1$ and $z \geq
0$, $(1+z/x)^x \geq 1+z$ and the last inequality follows since
$\sum_{i \in ALIVE_e} f^*_i \geq n_e$. It follows that the total
number of weight augmentation steps is at most:
$$
\log_2 (2gc)^\alpha = O(\alpha \log gc )
$$
\end{proof}

\begin{theorem}\label{frac_comp}
For the weighted case, the fractional algorithm is $O(\log
(mc))$-competitive. In case all the costs are equal to $1$, the
algorithm is $O(\log c)$-competitive.
\end{theorem}

\begin{proof}
The cost associated with the online algorithm is $\sum_{i \in REQ}
\min \{f_i,1\} p_i$, which we will denote by $C_{ON}$. Consider a
weight augmentation step performed for an edge $e$. In step
\ref{alg1} of the algorithm, the weights of at most $c+1$ requests
change from $0$ to $1/(gc)$. This is because before the current
request arrived, there could have been at most $c$ requests
containing the edge $e$ and having $f_i=0$ (the maximum capacity
is $c$). Since the maximum cost is $g$, the total increase of
$C_{ON}$ in step \ref{alg1} of the algorithm is at most
$(c+1)\frac{1}{gc}g = 1+1/c$. If follows that in step \ref{alg1},
the quantity $\sum_{i \in ALIVE_e} f_i$ can increase by at most
$1+1/c$. A weight augmentation is performed as long as $\sum_{i
\in ALIVE_e} f_i < n_e$. Before step \ref{alg2} we have that
$\sum_{i \in ALIVE_e} f_i < n_e + 1+1/c$. Thus, the total increase
of $C_{ON}$ in step \ref{alg2} of the algorithm does not exceed
$$
\sum_{i \in ALIVE_e} f_i p_i \frac{1}{n_e p_i} = \sum_{i \in
ALIVE_e} \frac{f_i}{n_e} < 2+1/c
$$
It follows that the total increase of $C_{ON}$ in a weight
augmentation step is at most $3+2/c$. Using lemma \ref{frac_steps}
which bounds the number of augmentation steps, we conclude that
the algorithm is $O(\log (gc))$-competitive.

For the weighted case, we saw that the input can be transformed so
that $g \leq 2mc$, which implies that the algorithm is $O(\log
(mc))$-competitive. In case all the costs are equal to $1$, $g$ is
also equal to $1$ and the algorithm is $O(\log c)$-competitive.

\end{proof}

\section{Randomized algorithm for \\ admission control}\label{sec:rand}

We describe in this section an $O(\log^2 (mc))$-competitive
randomized algorithm for the weighted case and a slightly better
$O(\log m \log c)$-competitive randomized algorithm for the
unweighted case.

The algorithm maintains a weight $f_i$ for each request $r_i$,
exactly like the fractional algorithm. Assume now that the
algorithm receives a request $r_i$ with cost $p_i$. The following
is performed in this case.
\begin{enumerate}
    \item Perform all the weight augmentations according to the fractional algorithm.
    \item Reject all requests whose weight is at least $\frac{1}{12 \log (mc)}$.\label{rand2}
    \item For every request $r$, if its weight $f$ increased by $\delta$, then reject the request $r$
    with probability $12\delta \log (mc)$.\label{rand1}
    \item If the current request $r_i$ cannot be accepted (some edge would be over capacity), then reject the
    request. Else, accept the request $r_i$.\label{rand3}
\end{enumerate}

We can assume that $|REQ_e|$, the total number of requests whose
paths contain a specific edge $e$, is less than $4mc^2$. To see
this, note that the fractional algorithm normalizes the costs so
that the minimum cost is $1$ and the maximum cost is at most
$2mc$. If $|REQ_e| \geq 4mc^2$, then since the optimal solution
can accept at most $c$ requests from $REQ_e$, it must pay a cost
of at least $t-2mc^2$ for requests rejected out of $REQ_e$, where
$t$ is the total cost of these requests. The online algorithm can
reject all the requests in $REQ_e$, pay $t$ and it will still be
$2$-competitive with respect to the requests in $REQ_e$, since $t
\geq 4mc^2$.

\begin{theorem}\label{rand_comp}
For the weighted case, the randomized algorithm is $O(\log^2
(mc))$-competitive.
\end{theorem}

\begin{proof}
Denote by $C_{frac}$ the cost of the fractional algorithm. The
expected cost of requests rejected in step \ref{rand1} of the
algorithm is at most $12 C_{frac} \log (mc)$. The cost of requests
rejected in step \ref{rand2} has the same upper bound.

We now calculate the probability for a request $r$ to be rejected
in step \ref{rand3}. This can happen only if the path of request
$r$ contains an edge $e$ for which $\sum_{i \in ALIVE_e} f_i \geq
n_e$ but the randomized algorithm rejected less than $n_e$
requests whose paths contain the edge $e$. All the requests with
weight at least $\frac{1}{12 \log (mc)}$ are rejected for sure,
so we can assume that $f_i < \frac{1}{12 \log (mc)}$ for all $i
\in ALIVE_e$.

Suppose that $i \in ALIVE_e$ and that during all runs of step
\ref{rand1} of the algorithm the request $r_i$ has been rejected
with probabilities $q_1,\ldots,q_n$, where $\sum_{k=1}^n q_k = 12
f_i \log (mc)$. The probability that $r_i$ will be rejected is at
least
$$
1-\prod_{k=1}^n (1-q_k) \geq 1-e^{-\sum_{i=k}^n q_k} = 1-e^{-12
f_i \log mc} \geq 6 f_i \log mc
$$
The last inequality follows since for all $0 \leq x \leq 1$,
$1-e^{-x} \geq x/2$.

The number of requests in $ALIVE_e$ which were rejected by the
algorithm is a random variable whose value is the sum of mutually
independent $\{0,1\}$-valued random variables and its expectation
is at least $\mu = 6 n_e \log mc$. By Chernoff bound (c.f., e.g.,
\cite{AS00}), the probability for this random variable to get a
value less than $(1 - a) \mu$ is at most $e^{-a^2 \mu/2}$ for
every $a > 0$. Therefore, the probability to be less than $n_e$ is
at most
$$
e^{-(1-\frac{1}{6 \log mc})^2(6 n_e \log mc)/2} \leq \frac{3}{m^3
c^3}
$$
The request costs were normalized, so that the maximum cost is at
most $2mc$. Each edge is contained in the paths of at most $4mc^2$
requests. Therefore, the expected cost of requests which are
rejected in step \ref{rand3} because of this edge is at most
$(4mc^2) (2mc) 3/(m^3c^3) = 24/m$. Thus, the total expected cost
of requests rejected in step \ref{rand3} is $24$. The result now
follows from Theorem \ref{frac_comp}.

\end{proof}

For the unweighted case we slightly change the algorithm as
follows. In step \ref{rand1} of the algorithm we reject a request
with probability $4\delta \log m$, and in step \ref{rand2} we
reject all the requests whose weight is at least $1 / (4 \log m)$.

\begin{theorem}\label{unweighted_rand_comp}
For the unweighted case, the randomized algorithm is $O(\log m
\log c)$-competitive.
\end{theorem}

\begin{proof}
Following the proof of Theorem \ref{rand_comp}, we get that the
probability for an edge to cause a specific request to be rejected
in step \ref{rand3} of the randomized algorithm is at most
$$
e^{-(1-\frac{1}{2 \log m})^2(2 n_e \log m)/2} \leq \frac{3}{m}
$$

Denote by $Q$ the quantity $max_{e \in E} (|REQ_e| - c_e)$. Hence,
$Q$ is the maximum excess capacity in the network. The total
expected cost of requests rejected in step \ref{rand3} is at most
$Q(3/m)m = 3Q$. It is obvious that the optimal solution must
reject at least $Q$ requests. The result now follows from Theorem
\ref{frac_comp}.

\end{proof}

\section{Reducing online set cover to admission control}\label{sec:reduction}
We now describe the reduction between online set cover and
admission control. Suppose we are given the following input to the
online set cover with repetitions problem: $X$ is a ground set of
$n$ elements and $\SS$ is a family of $m$ subsets of $X$, with a
positive cost $c_S$ associated with each $S \in \SS$. The instance
of the admission control to minimize rejections problem is
constructed as follows: The graph $G=(V,E)$ has an edge $e_j$ for
each element $j \in X$. The capacity of the edge $e_j$ is defined
to be the number of sets that contain the element $j$. The maximum
capacity is therefore at most $m$.

The requests are given to the admission control algorithm in two
phases. In the first phase, before any element is given to the
online set cover algorithm, we generate $m$ requests to the
admission control online algorithm. For every $S \in \SS$, the
request consists of all the edges $e_j$ such that $j \in S$. The
online algorithm can accept all the requests and this will cause
the edges to reach their full capacity.

In the second phase, each time the adversary gives an element $j$
to the online set cover algorithm, we generate a request which
consists of the one edge $e_j$ and give it to the admission
control algorithm. In case the request caused the edge $e_j$ to be
over capacity, the algorithm will have to reject one request in
order to keep the capacity constraint.

In case there is a feasible cover for the input given to the
online set cover problem, there is no reason for the admission
control algorithm to reject requests that were given in the second
phase. This is because requests in the second phase consist of
only one edge. Thus, we can assume that the admission control
algorithm rejects only requests given in the first phase, which
correspond to subsets of $X$.

It is easy to see that the requests rejected by the admission
control algorithm correspond to a legal set cover. We reduced an
online set cover problem with $n$ elements and $m$ sets to an
admission control problem with $n$ edges and maximum capacity at
most $m$. The fact that the requests we generated are not simple
paths in the graph can be easily fixed by adding extra edges.

\section{Deterministic bicriteria algorithm for online set cover}\label{sec:det}

In this section we describe, given any constant $\epsilon > 0$, an
$O(\log m \log n)$-competitive deterministic bicriteria algorithm
that covers each element by at least $(1-\epsilon)k$ sets, where
$k$ is the number of times the element has been requested, whereas
the optimum covers it $k$ times. We assume for simplicity that all
the sets have cost equal to 1. The result can be easily
generalized for the weighted case using techniques from
\cite{AAABN03}.

The algorithm maintains a weight $w_S > 0$ for each $S \in \SS$.
Initially $w_S = 1/(2m)$ for each $S \in \SS$. The weight of each
element $j \in X$ is defined as $w_j = \sum_{S \in \SS_j} w_S$,
where $\SS_j$ denotes the collection of sets containing element
$j$. Initially, the algorithm starts with the empty cover $\CC =
\emptyset$. For each $j \in X$, we define $cover_j = |\SS_j
\bigcap \CC|$, which is the number of times element $j$ is covered
so far. The following potential function is used throughout the
algorithm:

$$ \Phi = \sum_{j \in X} n^{2(w_j - cover_j)}$$

We give a high level description of a single iteration of the
algorithm in which the adversary gives an element $j$ and the
algorithm chooses sets that cover it. We denote by $k$ the number
of times that the element $j$ has been requested so far.

\begin{enumerate}
    \item If $cover_j \geq (1-\epsilon)k$, then do nothing.
    \item Else, while $cover_j < (1-\epsilon)k$, perform a {\em weight augmentation}:
    \begin{enumerate}
        \item For each $S \in \SS_j - \CC$, $w_S \gets w_S(1+\frac{1}{2k})$.\label{detalg1}
        \item Add to $\CC$ all the subsets for which $w_S \geq 1$.\label{detalg2}
        \item Choose from $\SS_j$ at most $2 \log n$ sets to $\CC$\label{detalg3}
        so that the value of the potential function $\Phi$ does
        not exceed its value before the weight augmentation.
    \end{enumerate}
\end{enumerate}

In the following we analyze the performance of the algorithm and
explain which sets to add to the cover $\CC$ in step \ref{detalg3}
of the algorithm. The cost of the optimal solution $\CC_{OPT}$ is
denoted by $\alpha$.

\begin{lemma}\label{det_steps}
The total number of weight augmentations performed during the
algorithm is at most $O(\alpha \log m)$.
\end{lemma}

\begin{proof}
Consider the following potential function:
$$ \Psi = \prod_{S \in \CC_{OPT}} w_S$$
We now show three properties of $\Psi$:
\begin{itemize}
    \item The initial value of the potential function is:
    $(2m)^{-\alpha}$.
    \item The potential function never exceeds $1.5^{\alpha}$.
    \item In each weight augmentation step, the potential function
    is multiplied by at least $2^{\epsilon/2}$.
\end{itemize}
The first two properties follow directly from the initial value
and from the fact that no request gets a weight of more than
$1.5$. Consider an iteration in which the adversary gives an
element $j$ for the $k$th time. Now suppose that a weight
augmentation is performed for element $j$. We must have that
$cover_j < (1-\epsilon)k$, which means that the online algorithm
has covered element $j$ less than $(1-\epsilon)k$ times. The
optimal solution $OPT$ covers element $j$ at least $k$ times,
which means that there are at least $\epsilon k$ subsets of $OPT$
containing $j$ which were not chosen yet. Thus, in step
\ref{detalg1} of the algorithm the potential function is
multiplied by at least:
$$
(1+\frac{1}{2k})^{\epsilon k} \geq 2^{\epsilon / 2}
$$
It follows that for fixed $\epsilon>0$ the total number of weight
augmentation steps is at most:

$$
\frac{\log (3m)^\alpha}{\log 2^{\epsilon / 2}} = O(\alpha \log m)
$$
\end{proof}

\begin{lemma}\label{det_poten}
Consider an iteration in which a weight augmentation is performed.
Let $\Phi_s$ and $\Phi_e$ be the values of the potential function
$\Phi$ before and after the iteration, respectively. Then, there
exist at most $2 \log n$ sets that can be added to $\CC$ during
the iteration such that $\Phi_e \leq \Phi_s$. Furthermore, the
value of the potential function never exceeds $n^2$.
\end{lemma}

\begin{proof}
The proof is by induction on the iterations of the algorithm.
Initially, the value of the potential function $\Phi$ is less than
$n \cdot n = n^2$. Suppose that in the iteration the adversary
gives element $j$ for the $k$th time. For each set $S \in \SS_j$,
let $w_S$ and $w_S + \delta_S$ denote the weight of $S$ before and
after the iteration, respectively. Define $\delta_j = \sum_{S \in
\SS_j} \delta_S$. By the induction hypothesis, we know that $2(w_j
- cover_j) < 2$, because otherwise $\Phi_s$ would have been
greater than $n^2$. Thus, $w_j < cover_j + 1 \leq \lfloor
(1-\epsilon)k \rfloor + 1 \leq k$. This means that $\delta_j \leq
k \cdot 1/(2k) = 1/2$.

We now explain which sets from $\SS_j$ are added to $\CC$.

Repeat $2 \log n$ times: choose at most one set from $\SS_j$ such
that each set $S \in \SS_j$ is chosen with probability
$2\delta_S$. (This can be implemented by choosing a number
uniformly at random in [0,1], since $2\delta_j \leq 1$.)

Consider an element $j' \in X$. Let the weight of $j'$ before the
iteration be $w_{j'}$ and let the weight after the iteration be
$w_{j'}+\delta_{j'}$. Element $j'$ contributes before the
iteration to the potential function the value $n^{2w_{j'}}$. In
each random choice, the probability that we do not choose a set
containing element $j'$ is $1 - 2\delta_{j'}$. The probability
that this happens in all the $2 \log n$ random choices is
therefore $(1 - 2\delta_{j'})^{2 \log n} \leq n^{-4\delta_{j'}}$.

Note that $\delta_{j'} \leq 1/2$. In case we choose a set
containing element $j'$, then $cover_{j'}$ will increase by at
least $1$ and the contribution of element $j'$ to the potential
function will be at most $n^{2(w_{j'}+\delta_{j'}-1)} \leq
n^{2w_{j'}-1}$. Therefore, the expected contribution of element
$j'$ to the potential function after the iteration is at most

\begin{eqnarray*}
n^{-4\delta{j'}}n^{2(w_{j'}+\delta_{j'})}+(1-n^{-4\delta{j'}})n^{2w_{j'}-1}\\
=n^{2w_{j'}} ( n^{-2\delta{j'}} + n^{-1} - n^{-4\delta{j'}-1} )
\leq n^{2w_{j'}}
\end{eqnarray*}

where to justify the last inequality, we prove that $f(x)=n^x +
n^{-1} - n^{2x-1} \leq 1$ for every $x \leq 0$. To show this we
note that $f(0)=1$ and $f'(x)= n^x \log n (1 - 2n^{x-1})$. This
implies that $f'(x) \geq 0$ for every $x \leq 0$. We can conclude
that $f(x) \leq 1$ for every $x \leq 0$, as needed.

By linearity of expectation it follows that $\textbf{Exp}[\Phi_e]
\leq \Phi_s$. Hence, there exists a choice of at most $2 \log n$
sets such that $\Phi_e \leq \Phi_s$. The choices of the various
sets $S$ to be added to $\CC$ can be done deterministically and
efficiently, by the method of conditional probabilities, c.f.,
e.g., \cite{AS00}, chapter 15. After each weight augmentation, we
can greedily add sets to $\CC$ one by one, making sure that the
potential function will decrease as much as possible after each
such choice.
\end{proof}

\begin{theorem}\label{det_comp}
The deterministic algorithm for online set cover is $O(\log m \log
n)$-competitive.
\end{theorem}

\begin{proof}
It follows from Lemma \ref{det_steps} that the number of
iterations is at most $O(\alpha \log m)$. By Lemma
\ref{det_poten}, in each iteration we choose at most $2 \log n$
sets to $\CC$ in step \ref{detalg3} of the algorithm. The sets
chosen is step \ref{detalg2} of the algorithm are those which have
weight at least $1$. The sum of weights of all the sets is
initially $1/2$ and it increases by at most $1/2$ in each weight
augmentation. This means that at the end of the algorithm, there
can be only $O(\alpha \log m)$ sets whose weight is at least $1$.
Therefore, the total number of sets chosen by the algorithm is as
claimed.
\end{proof}

\section{Concluding Remarks}\label{sec:conclude}

\begin{itemize}
    \item An interesting open problem is to decide if the algorithm presented here
    for the admission control problem can be derandomized.
    \item Recently, Feige and Korman established a lower bound of $\Omega(\log m \log n)$
    for the competitive ratio of any randomized polynomial time algorithm for the
    online set cover problem, under the $BPP \neq NP$ assumption. It is interesting to
    decide whether this lower bound applies for superpolynomial time algorithms as well.
    \item The algorithms we gave for the admission control problem did not use
    the fact that the requests are simple paths in the graph. All
    the algorithms treated a request as an arbitrary subset of
    edges.

\end{itemize}



\end{document}